\documentclass[letter,english,11pt]{article}
\def\BibTeX{{\rm B\kern-.05em{\sc i\kern-.025em b}\kern-.08emT\kern-.1667em\lower.7ex\hbox{E}\kern-.125emX}}

\usepackage[margin=1in,footskip=0.3in]{geometry}
\usepackage{verbatim}
\usepackage{amssymb,amsmath,amsthm,epsfig}
\usepackage{enumitem}
\usepackage{graphicx,wrapfig,lipsum}
\usepackage[font=small]{caption}
\usepackage{subcaption}
\usepackage[english]{babel}
\usepackage[utf8]{inputenc}
\usepackage{algorithm}
\usepackage[noend]{algpseudocode}
\usepackage{url}
\usepackage{color}
\usepackage[all]{xy}
\usepackage{rotating}
\usepackage{ifpdf}
\usepackage{tcolorbox}

\usepackage{titlesec}
\usepackage{sidecap}

\usepackage{xspace}

\newtheorem{theorem}{Theorem}
\newtheorem{lemma}{Lemma}

\newtheorem{invariant}{Invariant}
\newtheorem*{lemma*}{Lemma}
\newtheorem{corollary}{Corollary}

\newtheorem*{theorem*}{Theorem}
\newtheorem*{corollary*}{Corollary}

\newcommand{\cS}{{\mathcal{S}}}
\newcommand{\cF}{{\mathcal{F}}}
\newcommand{\cP}{{\mathcal{P}}}
\newcommand{\cX}{{\mathcal{X}}}
\newcommand{\cL}{{\mathcal{L}}}

\newcommand{\eat}[1]{}
\allowdisplaybreaks

\newcommand{\calS}{\mathcal{S}}
\newcommand{\calF}{\mathcal{F}}

\newcommand{\pr}[1]{{\rm Pr} \left[ #1 \right]}
\newcommand{\ex}[1]{{\rm E} \left[ #1 \right]}

\newcommand{\eps}{\varepsilon}

\definecolor{Darkblue}{rgb}{0,0,0.4}
\definecolor{Brown}{cmyk}{0,0.61,1.,0.60}
\definecolor{Purple}{cmyk}{0.45,0.86,0,0}
\definecolor{brickred}{rgb}{0.8, 0.25, 0.33}

\newcommand{\elemm}{MAP_{elem}}
\newcommand{\setm}{MAP_{set}}

\ifdefined\DEBUG
 \newcommand{\fab}[1]{\textcolor{red}{#1}}

  \def\rem#1{{\marginpar{\raggedright\scriptsize #1}}}
  \newcommand{\fabr}[1]{\rem{\textcolor{red}{$\bullet$ #1}}}
  
  \newcommand{\der}[1]{\rem{\textcolor{green}{$\bullet$ #1}}}

\else
  \newcommand{\fab}[1]{#1}
  \newcommand{\fabr}[1]{}

  \newcommand{\der}[1]{}

\fi

\usepackage{authblk}

\title{Dynamic Set Cover: Improved Algorithms and Lower Bounds}
\author[1]{Amir Abboud \thanks{amir.abboud@ibm.com}}
\author[2]{Raghavendra Addanki \thanks{raddanki@cs.umass.edu}}
\author[3]{Fabrizio Grandoni \thanks{fabrizio@idsia.ch. F. Grandoni is partially supported by the SNSF grants 200021$\_$159697$/$1 and 200020B$\_$ 182865$/$1.}}
\author[4]{Debmalya Panigrahi\thanks{debmalya@cs.duke.edu. D. Panigrahi is partially supported by NSF contracts CCF 1535972, CCF 1527084, an NSF CAREER Award CCF 1750140, and the Indo-US Virtual Networked Joint Center on Algorithms under Uncertainty.}}
\author[2]{Barna Saha \thanks{barna@cs.umass.edu. B. Saha is partially supported by an NSF CRII grant CCF 1464310, an NSF CAREER Award CCF 1652303, and an Alfred P. Sloan fellowship.}}
\affil[1]{IBM Almaden Research Center}
\affil[2]{University of Massachusetts Amherst}
\affil[3]{IDSIA, USI-SUPSI}
\affil[4]{Duke University}

\begin{document}
\date{}
\maketitle
\begin{abstract}
We give new upper and lower bounds for the {\em dynamic} set cover problem. First, we give a $(1+\eps) f$-approximation for fully dynamic set cover in $O(f^2\log n/\eps^5)$ (amortized) update time, for any $\epsilon > 0$, where $f$ is the maximum number of sets that an element belongs to. In the decremental setting, the update time can be improved to $O(f^2/\eps^5)$, while still obtaining an $(1+\eps) f$-approximation. These are the first algorithms that obtain an approximation factor linear in $f$ for dynamic set cover, thereby almost matching the best bounds known in the offline setting and improving upon the previous best approximation of $O(f^2)$ in the dynamic setting.

To complement our upper bounds, we also show that a linear dependence of the update time on $f$ is necessary unless we can tolerate much worse approximation factors. Using the recent distributed PCP-framework, we show that any dynamic set cover algorithm that has an amortized update time of $O(f^{1-\eps})$ must have an approximation factor that is $\Omega(n^\delta)$ for some constant $\delta>0$ under the Strong Exponential Time Hypothesis. 
\end{abstract}

%
%

\newpage
\section{Introduction}

Suppose, we need to solve a combinatorial optimization problem where the input to the problem changes over time. In such a dynamic setting, recomputing the solution from scratch after every update can be prohibitively time consuming, and it is natural to seek dynamic algorithms that provide faster updates. In the last few decades, efficient dynamic algorithms have been discovered for many combinatorial optimization problems, particularly in graphs such as shortest paths \cite{F85,di:04,ack:17,DBLP:conf/stoc/ItalianoKLS17}, connectivity \cite{HK99,HDT01,W17,ACK17}, maximal independent set and coloring \cite{BCHN18,aoss:18,DBLP:conf/icalp/OnakSSW18}. For many of these problems, maintaining exact solutions is prohibitively expensive under various complexity conjectures \cite{AW14,KPP16,AVY18,HKNS15}, and thus the best approximation bounds are sought. In their seminal work~\cite{or:10}, Onak and Rubinfeld proposed an algorithm for {\em matching} and {\em vertex cover} that maintains $O(1)$-approximate solutions to the maximum matching and minimum vertex cover in the graph. The algorithm runs in $t\cdot {\rm polylog}(n)$ time for any sequence of $t$ edge insertions and deletions in an $n$-vertex graph, i.e., in $O({\rm polylog}(n))$ time when {\em amortized} over all the updates. This has led to a flurry of activity in dynamic algorithms for matching and vertex cover \cite{s:16,ns:16,gp:13,bch:17,bhattacharya2016new,bhi:15,bs:16,bgs:15}, and more recently, for the more general set cover problem~\cite{bhi:icalp,gkkp:17,bch:17} that we study in this paper.

In the set cover problem, we are given a universe $X$ of $n$ elements  and a family $\calS$ of $m$ sets on these elements. The goal is to find a minimum-cardinality subfamily of sets $\calF \subseteq \calS$ such that  $\calF$ covers all the elements of $X$. The two traditional lines of inquiry for this problem are via greedy and primal dual algorithms, and have respectively led to a $\ln n$- and an $f$-approximation. Here, $f$ is the maximum number of sets that an element belongs to in the set system $\calS$. Both these results are known to be tight under appropriate complexity-theoretic assumptions~\cite{DS14,KR08}. In the dynamic setting, the set system $\calS$ is fixed, but the set of elements that needs to be covered in $X$ changes over time. In particular, after the {\em insertion} of a new element, or the {\em deletion} of an existing one, the solution has to be updated to maintain feasibility and the approximation guarantee. The time taken to perform these updates is called the {\em update time} of the algorithm, and is often stated amortized over any fixed prefix of updates. 


As in the case of the offline problem, dynamic algorithms for set cover have also followed two lines of inquiry.
The first is to use greedy-like techniques, which were recently shown to yield an $O(\log n)$-approximation in $O(f\log n)$ 
update time by Gupta, Kumar, Krishnaswamy, and Panigrahi~\cite{gkkp:17}.\footnote{All update times stated in this paper are amortized, unless stated otherwise.}
The second is to use a primal-dual framework, which was employed by Bhattacharya, Henzinger, and Italiano~\cite{bhi:icalp} 
to give an $O(f^2)$-approximation in $O(f\log{(m+ n)})$ update time. Gupta {\em et al.}~\cite{gkkp:17}
and Bhattacharya, Chakrabarty, and Henzinger~\cite{bch:17} also obtained a 
different but incomparable result using the primal-dual technique, which improves the update time to $O(f^2)$ thereby removing 
the dependence on $n$ and $m$, but at the cost of a weaker approximation bound of $O(f^3)$. What stands 
out in these results is that: 
\begin{itemize}[noitemsep]
\item While dynamic and offline approximation factors match
at $O(\log n)$, there is no $O(f)$-approximation known for the dynamic 
setting. Indeed, the only previous algorithm we are aware of that 
achieves this bound is one that recomputes the offline $f$-approximation 
after every update.
\item The update times of these algorithms depend on $\log n$ and $f$. While the 
dependence on $\log n$ is not required, at least if we settle for an $O(f^3)$ approximation \cite{gkkp:17,bch:17},
it is not clear if the polynomial dependence on $f$ is fundamental. For instance, 
might it be possible to design a dynamic set cover algorithm whose update time 
only has a logarithmic dependence on $f$?
\end{itemize}
\subsection{Our Results}
\label{sec:results}
Our first result closes the gap between offline and dynamic approximation for the set cover problem: 
{\bf for any $\eps > 0$, we give a $(1+\eps) f$-approximation algorithm for dynamic set cover with an update time of $O(f^2\log n/\eps^5)$}. 
Previous algorithms for dynamic set cover heavily rely on deterministically maintaining a greedy-like or primal-dual structure 
on the set cover solution. Instead, our algorithm is based on the observation that 
a simple offline algorithm for the set cover problem achieves a $(1 +\eps) f$-approximation in $O(f/\eps)$
expected update time when the elements are deleted in a random order. We switch this statement
around by transferring the randomness to the algorithm in order to handle an arbitrary
sequence of deletions (and insertions). As a result, our algorithm is randomized, and our update
time bound holds in expectation. (The approximation bound holds deterministically.)

In the {\em decremental} setting where elements can only be deleted but not inserted,
a simplification of the above algorithm yields the same {\bf approximation factor of $(1+\eps) f$ in amortized update time $O(f^2/\eps^5)$}. This can be compared with the result of Gupta {\em et al.} \cite{gkkp:17} which achieves a (larger) $O(f^3)$-approximation with (roughly) the same update time, but in the fully dynamic case. As far as we know, the approximation bounds of \cite{gkkp:17,bch:17} do not change when considering the decremental setting, which has been extensively studied in the past for other problems~\cite{henzinger2015unifying,DBLP:conf/stoc/ItalianoKLS17,DBLP:conf/focs/ChechikHILP16}.\footnote{For the incremental setting, where elements can be inserted but not deleted, the offline set cover algorithm itself gives an $f$-approximation in $O(f)$ update time.}

Finally, we turn to the problem of determining the dependence of the update time on $f$. Using the recently introduced framework of distributed PCP \cite{ARW17} from 
fine-grained complexity theory, {\bf we show that under the {\em Strong Exponential Time Hypothesis (SETH)}, any dynamic 
set cover algorithm that has an (amortized) update time of $O(f^{1-\eps})$ for any fixed $\eps > 0$ must have an approximation
factor of $(n/\log f)^{\Omega(1)}$}. Since a polynomial dependence on $n$ in the approximation factor is rather weak, this
result essentially states that any dynamic set cover algorithm must have a linear dependence on $f$ in the update time\footnote{In our model, the elements and sets are fixed, as well as their membership relations, and the updates can change which of the elements are ``active'' in the instance. Thus, an element insertion can be specified with $O(\log{n})$ bits, rather than the $\Omega(f)$ bits that may be required to list all its membership relations. This makes an $\Omega(f)$ lower bound on the update time non-trivial.}. 
This shows the update time bound of \cite{gkkp:17} to achieve $O(\log{n})$ approximation is essentially tight within a $\log{n}$ factor. Our lower bound holds even if the algorithm is allowed a preprocessing stage with arbitrary polynomial runtime, and it also applies to the set-updates model where the elements are fixed but sets get inserted and deleted. This model is much more popular in the streaming setting \cite{MV16,AK18}, especially when there are only insertions (see the work by Indyk et al. \cite{I+17} and the many references therein).
This is a novel application of the growing area of fine-grained complexity theory to show hardness of approximation of an NP-Hard problem.
\subsection{Our Techniques}
\label{section:preliminary}
A natural starting point for our work is to use the deterministic greedy or primal dual techniques for dynamic set cover from \cite{bhi:icalp,gkkp:17,bch:17}. An alternative strategy is to generalize previous randomized approaches for dynamic vertex cover \cite{bgs:15,s:16}. At a very high level, all these algorithms derive their results from maintaining, either explicitly or implicitly, a very structured dual solution that lower bounds the cost of the algorithm. Indeed, the algorithm of \cite{gkkp:17} for dynamic set cover can be thought of as a derandomization of the dynamic vertex cover algorithm of \cite{s:16}. In order to improve the approximation factor to $O(f)$, these dual solutions must only violate the dual packing constraints by a constant factor (as against an $\Omega(f)$ violation in the previous results), but this requirement is too strict for the analysis framework of these papers that effectively rely on an $f$-discretization of the dual space. 

Hence, we need a significantly new approach to improve the approximation factor to $O(f)$. We start with the following folklore algorithm for {\em offline} set cover. Initially, all elements are uncovered and the algorithm has an empty solution. Pick an arbitrary uncovered element and call it a pivot $p$. Then, include all sets containing $p$ in the solution and mark all elements in those sets as covered. Repeat this process until all elements get covered. This algorithm runs in $O(nf)$ time and achieves an $f$-approximation, since no two pivots share a common set and the algorithm picks at most $f$ sets for each pivot. We call this the {\em deterministic covering} algorithm.

Now, consider a decremental setting where elements are deleted over time, but in a {\em uniform random} order. 
A small modification to the deterministic covering algorithm gives a $(1+ O(\eps)) f$-approximation in $O(f/\eps)$ update time in this setting. Initially, we run deterministic covering to produce a feasible cover. 
During the deletion phase, the approximation bound may no longer hold because pivots are being deleted. To restore the bound, we re-run the deterministic covering algorithm whenever an $\eps$-fraction of the pivots have been deleted. Since the number of undeleted pivots forms a lower bound on the optimal solution, it follows that this algorithm maintains an $f/(1-\eps) = (1+O(\eps)) f$-approximation. 

Let us now consider the update time. Clearly, deterministic covering takes $O(nf)$ time in every run. So, the question is: how frequently do we run it? Because of the random deletion order, we expect to delete an $\eps$-fraction of all elements before an $\eps$-fraction of pivots gets deleted. This suggests an informal amortized update bound of $O(nf/(\eps n)) = O(f/\eps)$. It turns out that this informal idea can be made formal, but we skip the details here since we are going to use this only for intuitive purposes.

More interesting for us is to transition from a random deletion order to an adversarial deletion order. The same update rule gives a $(1+\eps) f$-approximation, but now, the bound on update time may no longer hold. For instance, if all the pivots are deleted before other elements, the amortized update time is clearly much higher when the first $\eps$-fraction of pivots gets deleted. 

Our main idea, at this juncture, is to transfer the randomization from the deletion sequence to the algorithm itself. More specifically, instead of picking a pivot arbitrarily from the uncovered elements in each step, let us select it uniformly at random. We call this the {\em random covering} algorithm. Our hope is that an (oblivious) adversary deleting a single element will be able to pick a specific pivot with probability no higher than $1/n$. This would ensure that in expectation, an  $\eps$-fraction of the elements will have to be deleted before an $\eps$-fraction of pivots is, as in the random deletion scenario. 

However, this intuition is not quite correct. While the first pivot is indeed uniformly distributed over all elements, the subsequent pivots are not. To see this, consider the following example: suppose the sets represent edges of a graph containing $(n-2)/f$ cliques on $f$ vertices each, and an isolated edge. For a vertex on the isolated edge to be chosen as the second pivot, it must not be covered by sets containing the first pivot and should be selected as the second pivot; the probability for this event is given by: $(1 - 2/n)\cdot (1/(n-f))$. Clearly, this probability exceeds $1/n$ for $f > 2$.
As a consequence, the expected number of element deletions after which we need to run random covering might be smaller than $\eps n$. To overcome this bottleneck, we employ a more fine-grained update procedure: instead of running random covering over the entire undeleted instance, we run it only for a subset of elements. We maintain sufficient structure in the solution to still claim a $(1+\eps) f$-approximation, while improving the update time to $O(f^2/\eps^5)$ for the decremental setting, and $O(f^2\log n/\eps^5)$ for the fully dynamic setting where elements can be inserted in addition to deletions.

\eat{

\smallskip\noindent {\bf Roadmap.}
We present the random covering algorithm for the decremental setting in Section~\ref{sec:decremental}. In Section~\ref{sec:fullyDynamic}, we further generalize these ideas to the fully dynamic setting, albeit at a slightly worse update time. Finally in Section~\ref{sec:lb}, we describe our lower bound results. 

}
\section{The Decremental Set Cover Algorithm}
\label{sec:decremental}

In this section, we give a dynamic set cover algorithm for the decremental setting.
We denote the initial set system by $(X, \cS)$, where $\cS = \{S_1, S_2, \ldots, S_m\}$
is a collection of subsets of the ground set $X$ that contains $n$ elements. The 
maximum number of subsets that an element belongs to is denoted $f$:
$$f = \max_{x\in X} |\{i: x\in S_i\}|.$$
The elements are deleted in a fixed sequence, independent of the randomness of 
the algorithm, that is represented by $X = \{x_1, x_2, \ldots, x_n\}$.

\subsection{The Algorithm}
The description of the algorithm comprises two phases: the {\em initial phase}
where the algorithm selects a feasible solution at the outset, and the {\em update phases}
where the algorithm changes its solution in response to the deletion of elements. 
The feasible solution that the algorithm maintains dynamically
is denoted by $\cF$. Recall that the goal is to ensure that the cost of $\cF$ is at 
most $(1+\eps)f$ times that of an optimal solution for the set of undeleted 
elements at all times. 

Both the initial and the update phases use a common subroutine that we call 
the {\em random cover} subroutine. We describe this subroutine first.

\smallskip\noindent{\bf The Random Cover Subroutine.}
The random cover subroutine takes as input a set system $(X', \cS')$ and outputs
a feasible set cover solution $\cF'$ for this set system. The algorithm is iterative,
where each iteration starts with a set of {\em uncovered} elements $Y\subseteq X'$,
adds a collection of sets $\cF^+ \subseteq \cS'$ to the solution $\cF'$, and 
removes all the elements covered by the sets in $\cF^+$ from the set of the uncovered 
elements $Y$ for the next iteration. Initially, all elements in $X'$ are uncovered, 
i.e., $Y = X'$, and the solution $\cF'$ is empty, i.e., $\cF' = \emptyset$. It only 
remains to describe an iteration, or more precisely, the sets $\cF^+$ added to 
the solution $\cF'$ in an iteration. The selection of $\cF^+$ has three steps.
First, the algorithm picks the set in $\cS'$ that covers the maximum number 
of uncovered elements, breaking ties arbitrarily. Let us call this set $Z$,
i.e., $Z = \arg\max_{S\in \cS'} |S \cap Y|$. Next, the algorithm chooses an 
element in $Z \cap Y$, i.e., an uncovered element in the chosen set, {\em uniformly
at random}, and calls this element the {\em pivot} for the current iteration. 
Let us call this pivot $p \in_{\rm u.a.r.} Z\cap Y$. Finally, all sets in $\cS'$ that 
contain the pivot are added to the solution, i.e., $\cF^+ = \{S\in \cS': p\in S\}$. 
The random cover subroutine ends when all elements in $X'$ are covered by the 
solution  $\cF'$, i.e., $Y = \emptyset$. This algorithm can be implemented in $O(f|X'|)$ 
deterministic time (details in Section \ref{sec:running-time}).

The above completes the description of the random cover subroutine. 
However, it will be convenient to introduce some additional notation 
for this process that we will use later. 
Each iteration is characterized by its pivot $p$. 
We map the pivot to the set $S(p):= Z\cap Y$ from which it is chosen. 
If $|S(p)| \in [2^i, 2^{i+1})$, we say that $p$ is a
level-$i$ pivot, and denote $\ell(p) = i$. Note that by the definition of the 
random cover subroutine, the pivots chosen in successive iterations have 
monotonically non-increasing levels, i.e., if pivot $p$ is chosen in an 
earlier iteration and pivot $p'$ in a later iteration, then 
$\ell(p) \geq \ell(p')$. Finally, if the sets $\cF^+$ are added to the 
solution $\cF'$ in an iteration with pivot $p$, then we denote 
$\cF(p) = \cF^+$. The set of previously uncovered elements that are 
covered by $\cF(p)$ is denoted $\cX(p)$.

\smallskip\noindent{\bf Initial Phase.}
In the initial phase, the random cover subroutine is run on the the entire
input set system $(X, \cS)$. This produces the initial solution $\cF$.

In the algorithm, we also maintain sets $P$, $D$, and $U$ that respectively 
represent {\em all}, {\em deleted}, and {\em undeleted} pivots. At the end 
of the initial phase, all the pivots in $\cF$ are added to $P$
and $U$, and $D$ is empty. When an element $e$ is deleted, if $e$ is in
$P$, then we move $e$ from $U$ to $D$, i.e., change its status from
undeleted to deleted. Importantly, we keep this element in $P$. 
Changes to $P$ are done only at the end of an update phase that 
we describe below.

\smallskip\noindent{\bf Update Phase.}
An update phase is triggered when the number of deleted pivots
exceeds an $\eps$-fraction of the total number of pivots,
i.e., $|D| \geq \eps\cdot |P|$. In an update phase, 
the algorithm first fixes a level $\ell$ using a process that we describe 
later called the {\em level fixing process}. Having fixed this level $\ell$,
the algorithm discards all sets $\cF(p)$ from $\cF$ that were added by 
pivots $p$ at levels $\ell$ or lower, i.e., where $\ell(p) \leq \ell$.
Correspondingly, these pivots $p$ are also removed from $P$ and from 
either $D$ or $U$ depending on whether they are deleted or undeleted.
As a result of this change to $\cF$, some elements become uncovered in $\cF$;
this set is denoted by
$X'$. The algorithm now runs the Random Cover subroutine on the instance 
$(X', \cS')$ induced by $X'$, where $\cS'=\{S\cap X':S\in \cS, S\cap X'\neq \emptyset\}$. The resulting sets $\cF'$ are added to the overall solution $\cF$.
Correspondingly, the newly selected pivots are also added to $P$ and $U$.
We say that levels $\ell$ and below have been updated in the current update phase. (Note that the newly selected pivots will be at level $\ell$ or below.)
Clearly, this restores feasibility of the solution $\cF$. We already argued that the call to the Random Cover subroutine can be performed in $O(f|X'|)$ deterministic time. The same upper bound holds for the remaining operations related to the construction of the instance $(X',\cS')$ and to the update of the approximate solution (see Section \ref{sec:running-time} for the details).

\smallskip\noindent{\bf The level fixing process.}
We now describe the level fixing process. Let $\{0,1,\ldots,L=\lfloor \log_2 n \rfloor\}$ be the set of levels. 
Let $P_j$, $D_j$, and $U_j$ respectively denote the current total set of pivots, deleted pivots, and undeleted pivots
at a given level $j$. 
This process finds a level $\ell$
with the following property: for every level $i \leq \ell$, $\sum_{j=i}^{\ell}|D_j|\geq \eps \cdot \sum_{j=i}^{\ell}|P_j|$.
In other words, the fraction of deleted pivots in levels $i, i+1, \ldots, \ell$ is at least an 
$\eps$-fraction of the total number of pivots in these levels. 
We say that level $\ell$ is {\em critical}. The next lemma claims that at least one critical level 
exists whenever the number of deleted pivots is an $\eps$-fraction 
of the total number of pivots.

\newcommand{\fail}{\text{\sc fail}\xspace}

\begin{lemma}
\label{lem:criticalExists}
    If $\sum_{j=0}^{L}|D_j|\geq \eps\cdot \sum_{j=0}^{L}|P_j|$, then there exists at least one critical level.
\end{lemma}
\begin{proof}
Suppose $\sum_{j=0}^{L}|D_j|\geq \eps\cdot \sum_{j=0}^{L}|P_j|$. Assume by contradiction that the claim is not true. Hence for each level $\ell$, there exists a level $\fail(\ell)\leq \ell$ (in case of ties, take the lowest such level) such that the condition does not hold, namely $\sum_{j=\fail(\ell)}^{\ell}|D_j|<\eps\cdot \sum_{j=\fail(\ell)}^{\ell}|P_j|$. We next define a sequence of levels $\ell_1,\ldots,\ell_q$ as follows. Set $\ell_1=L$. Given $\ell_i$, halt if $\fail(\ell_i)=0$, else set $\ell_{i+1}=\fail(\ell_i)-1$ and continue with $\ell_{i+1}$. Observe that the intervals $[\fail(\ell_i),\ell_i]$ are disjoint and span $[0,L]$. We have
$$
\sum_{j=0}^{L} |D_j| =\sum_{i=1}^{q}\sum_{j=\fail(\ell_i)}^{\ell_i}|D_j|< \sum_{i=1}^{q}\eps\cdot \sum_{j=\fail(\ell_i)}^{\ell_i}|P_j| =  \eps\cdot\sum_{j=0}^{L} |P_j|.
$$
This contradicts the assumption.
\end{proof}

The next lemma, shown in Section \ref{sec:running-time}, establishes the time complexity of the above algorithm.


\begin{lemma}
\label{lma:update-runtime}
    Suppose we perform an update at critical level $\ell$. Let $p_i$ be the total number of pivots at levels $i\leq \ell$ right before this update. Then, the total time taken for this update phase is $O(\sum_{i\leq \ell}f^2 p_i 2^i)$.
\end{lemma}

%
%
%
%

\subsection{Analysis of the Competitive Ratio}

\begin{lemma}
\label{lem:apxDec}
The competitive ratio of the algorithm is at most $f/(1-\eps)$.
\end{lemma}
\begin{proof}
Consider the data structure right before the $t$-th deletion. Let $P^t$ be the total set of pivots and $U^t$ the set of undeleted pivots at that time. We also let $OPT^t$ and $\cF^t$ be the optimal and approximate solution at that time. 

Observe that $|\cF^t|\leq f\cdot |P^t|$ by construction. We claim that $|OPT^t|\geq |U^t|$. This implies the claim since by construction $|P^t|\leq |U^t|/(1-\eps)$ at any time.

To see that, let us show by a simple induction that, for any two distinct $p,p'\in U^t$, there is no set $S\in \cF$ covering both $p$ and $p'$. 
Thus $OPT^t$ needs to include a distinct set for each element of $U^t$. The Random Cover subroutine applied to $X'$ never selects a pivot $p'$ that is covered by sets selected due to a previous pivot $p$. This implies that the property holds after the initialization step, where $X'=X$. 

Assume the property holds up to step $t\geq 1$, and suppose that at step $t+1$ an update happens at critical level $\ell$, involving elements $X'$. By inductive hypothesis and the properties of Random Cover, the claim holds for any pair of pivots $p,p'$ that are both contained in $X'$ or in its complement $X\setminus X'$. Furthermore by construction $X'$ is disjoint from any set $S$ that covers a pivot $p\in X\setminus X'$, hence $S$ cannot cover any pivot $p'\in X'$. 
\eat{
 By construction, the Random Cover subroutine applied to $X'$ never selects a pivot $p'$ that is covered by sets selected due to a previous pivot $p$. This implies that the property holds after the initialization step, where $X'=X$. Assume the property holds up to step $t\geq 1$, and consider step $t+1$. If the pivots are not updated during step $t$, there is nothing to show. Otherwise the algorithm will keep a proper prefix (in order of insertion) $\cF_{pref}^t$ of the sets $\cF^t$, and replace the remaining suffix $\cF_{suf}^t$ with $\cF_{suf}^{t+1}$ by means of the Random Cover subroutine. Let $X_{suf}^{t+1}$ be the undeleted elements at the beginning of step $t+1$ covered by $\cF_{suf}^{t+1}$, and $X_{pref}^{t+1}$ be the remaining undeleted elements at the same time. By the same argument as above, the property holds at step $t+1$ for any $p,p'\in X_{suf}^{t+1}$. It also holds for any $p,p'\in X_{pref}^{t+1}$ by inductive hypothesis. If $p\in X_{pref}^{t+1}$ and $p'\in X_{suf}^{t+1}$, by construction $p$ was selected before any pivot in $X_{suf}^t$. Thus the sets $\cup_{S\in \cF(p)}S$ and $X_{suf}^{t}$ (hence $X_{suf}^{t+1}\subseteq X_{suf}^{t}$) are disjoint.}
\end{proof}

\subsection{Analysis of the Amortized Update Time}

Our goal is to show that, after $t$ deletions, the 
expected time taken by the algorithm is $O\left(fn+\frac{f^2}{\eps^5}\cdot t\right)$. 
In particular, over a sequence of $n$ deletions, the expected amortized cost per deletion 
is $O\left(\frac{f^2}{\eps^5}\right)$.

Suppose we perform an update at critical level $\ell$. Let $P_i$ be the total set of pivots at level $i\leq \ell$ right before this update, of which $D_i$ denotes the set of deleted pivots. Let also $p_i=|P_i|$ and $d_i=|D_i|$. Recall that by Lemma~\ref{lma:update-runtime}, the total time taken for this update phase is $O(\sum_{i\leq \ell}f^2 p_i 2^i)$. We call a level $i\leq \ell$ \emph{charged} if $d_i\geq \frac{\eps}{2}p_i$ and \emph{uncharged} otherwise. We denote by ${\mathcal{L}}=(\ell_q,\ell_{q-1},\ldots,\ell_1)$ the decreasingly ordered sequence of charged levels $j\leq \ell$. Observe that, by the definition of critical level, $\ell_q=\ell$. Also let $D$ be the set of deleted pivots in the charged levels. The following lemma creates a useful mapping between $D$ and $P:=\cup_{i\leq \ell}P_i$.
\begin{lemma}\label{lem:matching}
There exists a $b$-matching $M$ between $D$ and $P$ such that:
    \begin{itemize}\itemsep0pt
        \item Each element of $P$ is matched to exactly one element of $D$ and each element of $D$ to at most $b = 2/\eps$ elements of $P$;
        \item If $d\in D$ is matched to $p\in P$, then $\ell(d)\geq \ell(p)$.
     \end{itemize}
\end{lemma}
\begin{proof}
Let us define $\ell_{0}=0$. For every $k=1,\ldots,q$, by definition  
\begin{equation}\label{eqn:lem:matching}
\sum_{j=\ell_{k-1}+1}^{\ell_k-1}d_j\leq \frac{\eps}{2}\sum_{j=\ell_{k-1}+1}^{\ell_k-1}p_j.
\end{equation}
Let us define ${\mathcal{P}}_{\ell_k}:=\cup_{j=\ell_{k-1}+1}^{\ell_k}P_j$. Therefore, for every $h=1,\ldots,q$, we have
\begin{align}
\sum_{k=h}^{q}d_{\ell_k} 
& =\sum_{j=\ell_{h-1}+1}^{\ell}d_j-\sum_{k=h}^{q}\sum_{j=\ell_{k-1}+1}^{\ell_k-1}d_j \nonumber\\
& \overset{\text{$\ell$ critical}}{\geq} \sum_{j=\ell_{h-1}+1}^{\ell}\eps  p_j-\sum_{k=h}^{q}\sum_{j=\ell_{k-1}+1}^{\ell_k-1}d_j\nonumber\\
& \overset{\eqref{eqn:lem:matching}}{\geq} \sum_{j=\ell_{h-1}+1}^{\ell}\eps p_j-\frac{\eps}{2}\sum_{k=h}^{q}\sum_{j=\ell_{k-1}+1}^{\ell_k-1}p_j \nonumber\\
& \geq  \frac{\eps}{2}\sum_{j=\ell_{h-1}+1}^{\ell}p_j = \frac{\eps}{2}\sum_{k=h}^{q}|{\mathcal{P}}_{\ell_k}|.\label{eqn:refinedMatching}
\end{align}


Let us replace each pivot $d\in D_{j}$, $j\in \cL$, with $2/\eps$ copies, and let us call the new set $D'_{j}$. Each copy of a pivot inherits the level of the original element. Let us sort $D':=\cup_{j\in \cL}D'_j$ in non-increasing order of level (breaking ties arbitrarily), and similarly sort $P$. Now, for every element $p\in P$ according to this order, we match $p$ with the first unmatched element $d'\in D'$. The $b$-matching is obtained by collapsing the copies in $D'$ of the same pivot in $D$. Observe that this allows us to match all elements of $P$ since
$$
|D'|=\frac{2}{\eps}|D|=\frac{2}{\eps}\sum_{k=1}^{q}d_{\ell_k}\overset{\eqref{eqn:refinedMatching}\text{ with $h=1$}}{\geq} \frac{2}{\eps}\cdot \frac{\eps}{2}\sum_{k=1}^{q}|{\mathcal{P}}_{\ell_k}|=|P|.
$$
It remains to show that the condition on the levels is satisfied. Suppose there exists some $d'\in D'$ matched to $p\in P$ with $\ell_{h-1}=\ell(d')<\ell(p)$. By construction this implies that 
$$
\frac{2}{\eps}\sum_{k=h}^{q}d_{\ell_k}=\sum_{k=h}^{q}|D'_{\ell_k}|<\sum_{k=h}^{q}|\cP_{\ell_k}|,
$$ 
which contradicts \eqref{eqn:refinedMatching}.
\end{proof}

Consider a sequence of $t$ deletions, and let $T$ be time right before the $(t+1)$-st deletion occurs (or the end of the execution if no such deletion exists). We wish to bound the expected total update time till time $T$ as a multiple of $t$. 
To that aim, we need a more global notation. Suppose that level $i$ is updated $q_i$ times in total, and let $P^j_i$ be the total set of pivots in the $j$-th such update. 
We use $P(i)$ to denote the multiset of pivots given by the union of the sets $P^j_i$, and define $p(i):=|P(i)|$. The pivots of type $P^j_i$ where level $i$ is charged on the $j$-th update are called charged and denoted $CH(i)$, $ch(i):=|CH(i)|$. 
The charged deleted pivots at level $i$ are denoted by $D(i)$, $d(i):=|D(i)|$.

We call level $i$ \emph{globally-charged} iff $ch(i)\geq \frac{\eps}{\fab{4}}p(i)$, and \emph{globally-uncharged} otherwise. We also let $GC$ denote the (random) set of globally-charged levels. We next show that, in order to bound the total update time, we can focus on globally-charged levels only.

\begin{lemma}\label{lem:timeGC}
The expected running time of the algorithm is
$$
O(fn)+O\left(\sum_{i=0}^{L} \pr{i\in GC}\cdot \ex{\frac{f^2}{\eps}2^i d(i) \;\middle\vert\; i\in GC}\right).
$$
\end{lemma}
\begin{proof}
Excluding the initialization cost of $O(fn)$ and by Lemma \ref{lma:update-runtime}, we can focus on bounding $O(\sum_i f^2 p(i) 2^i)$. We use the following token argument to upper bound the latter cost. We provide $f^2 2^i$ tokens to each pivot $p\in P(i)$, where each token can pay for a large enough constant amount of work. Then we transfer these tokens to charged deleted pivots in globally-charged levels so that all tokens are transferred and each charged deleted pivot at level $i$ is charged with at most $\frac{4f^2}{\eps}2^i$ tokens.

We next describe the transfer process. Let $M$ denote the $b$-matchings in Lemma \ref{lem:matching}. In particular, each pivot $p\in P(i)$ is matched with some charged deleted pivot $M(p)$ at no lower level, and each charged deleted pivot is matched with at most $2/\eps$ pivots. We remark that uncharged pivots $p$ have their $M(p)$ at a strictly higher level by construction. First of all, each pivot $p\in P(i)$ transfers its token\fab{s} to the corresponding charged deleted pivot according to $M$. Note that at this point each charged deleted pivot at level $i$ owns at most $\frac{2f^2}{\eps}2^i$ tokens.  

Next we proceed in increasing order of level $i$. For a given level $i$, each charged deleted pivot $d\in D(i)$ owns the tokens originally owned by $d$ and possibly tokens transferred from lower levels. If $i\in GC$ we do nothing. Otherwise (i.e., $i\notin GC$), we define a $b$-matching $M^i$ where each charged deleted pivot $d\in D(i)$ is matched with $\frac{2}{\eps}$ distinct uncharged pivots in $P(i)$ so that no uncharged pivot in $P(i)$ is matched twice. Note that this is possible since, by definition of globally-uncharged level, the uncharged pivots in $P(i)$ are at least 
$$
p(i)-ch(i)
\geq \left(1-\frac{\eps}{4}\right)p(i)
\geq \left(1-\frac{\eps}{4}\right)\frac{4}{\eps}ch(i)
\geq \frac{2}{\eps}ch(i), 
$$
where we assumed $\eps\leq 2$ w.l.o.g.
Now $d$ transfers an $\frac{\eps}{2}$-fraction of its tokens to each corresponding pivot in $M^i(d)$. Finally each matched uncharged pivot $p\in P(i)$ transfers the received tokens to the corresponding deleted pivot $M(p)$ in the global $b$-matching $M$. Observe that $M(p)$ must be at strictly higher level than $p$, hence the process is well-defined.

Clearly at the end of the process all the tokens are transferred to charged deleted pivots in globally-charged levels and no token is left. We next prove by induction that, at the end of iteration $i$ (where level $i$ is considered), the number of tokens charged to each $d\in D(i)$ is at most $\frac{4}{\eps}f^2 2^i$.
The claim follows. 

The base of the induction $i=0$ is trivially true. Indeed, if $i\notin GC$ all the tokens of $d$ are transferred to higher levels. Otherwise $d$ can only be charged with the starting number of tokens, which is at most $\frac{2f^2}{\eps}$ since there are no lower levels that can transfer tokens to level $i$. Note that the number of tokens that $d$ is charged with does not change in the rest of the token transfer process.

Next consider a level $i>0$, and assume the claim is true for levels $i-1$ and lower. For any $d\in D(i)$, again the claim holds trivially if $i\notin GC$. Otherwise, $d$ initially has up to $\frac{2f^2}{\eps}2^i$ tokens. Furthermore, $d$ can receive extra tokens from up to $\frac{2}{\eps}$ pivots $p$ of strictly lower levels. Each such $p$ at level $\ell\leq i-1$ transfers to $d$ an $\frac{\eps}{2}$-fraction of the tokens of some charged deleted pivot of level $\ell$. By the inductive hypothesis, the total number of tokens received by $d$ at the end of iteration $i$ is at most 
$$
\frac{2}{\eps}f^2 2^{i}+\frac{2}{\eps}\cdot \frac{\eps}{2} \cdot \frac{4}{\eps}f^2 2^{i-1}=\frac{4}{\eps}f^2 2^{i}.
$$
Again, the number of tokens that $d$ is charged with does not change in the rest of the token transfer process.
\end{proof}

Based on the above lemma, what remains to show is a bound for $\ex{2^i d(i) \;\middle\vert\; i\in GC}$. 
We bound this in terms of $\ex{t(i) \;\middle\vert\; i\in GC}$, where $t(i)$ is the (random) number of deletions that happen at level $i$. 
Instead of considering $t(i)$ directly, we rather focus on the following quantity. For a pivot $p(S)$ sampled from some set $S$ (considering only the uncovered elements at that time), let $i(S)$ be the relative position of $p(S)$ in $S$ w.r.t. the deletion order. We remark that $i(S)$ is uniformly distributed in $\{1,\ldots,|S|\}$. Define $x(i):=\sum_{S:p(S)\in D(i)}i(S)$.  Notice that deterministically $x(i)\leq t(i)$ since all the elements that appear in a set $S$ no later than the respective pivot $p(S)$ in the deletion order are deleted assuming $p(S)$ is deleted.

Let us also condition on $p(i)=p$ for some fixed value $p$, and consider $\ex{x(i) \;\middle\vert\; i\in GC,p(i)=p}$. We now relate this quantity to another random process. Suppose there is an adversary that defines a collection of exactly $p$ sets $\cS$ (possibly with repetition), where each set $S\in \cS$ has size $|S|\in [2^i,2^{i+1})$, and a deletion sequence over the elements of the sets. Now we sample a pivot $\tilde{p}(S)$ uniformly at random in each set $S\in \cS$, and let $\tilde{i}(S)$ be the relative position of the pivot $\tilde{p}(S)$ in $S$ w.r.t. the deletion sequence. The adversary is informed about the values $\tilde{i}(S)$. The adversary chooses a subcollection $\cS'\subseteq \cS$ of size at least $\frac{\eps^2}{8} p$, and computes $\tilde{x}_p(i)=\sum_{S\in \cS'}\tilde{i}(S)$. The adversary makes these choices in order to minimize $\ex{\tilde{x}_p(i)}$. 
\begin{lemma}\label{lem:domination}
$\ex{\tilde{x}_p(i)} \leq \ex{x(i) \;\middle\vert\; i\in GC,p(i)=p}$.
\end{lemma}
\begin{proof}
We use a coupling argument. Intuitively, the adversary can mimic the behavior of any execution of our decremental algorithm. In more detail, consider any execution of the decremental algorithm such that $i\in GC$ and $p(i)=p$. We couple the behavior of the adversary with this execution as follows. The adversary selects the same deletion order as in the input, and as collection $\cS$ precisely the sets that appear at level $i$ right before each update phase that involves that level (hence $|\cS|=p$). By coupling, we can assume that the sampled pivots in $\cS$ are precisely the pivots $P(i)$ of level $i$ in the execution of the algorithm. The collection $\cS'$ is given by the sets $S\in \cS$ such that the corresponding pivots are deleted and charged in the considered execution of the algorithm. Observe that $d(i)\geq \frac{\eps}{2}ch(i)\geq \frac{\eps^2}{8}p(i)=\frac{\eps^2}{8}p$, hence the constraint $|\cS'|\geq \frac{\eps^2}{8}p$ is satisfied. One has $\tilde{x}_p(i)=x(i)$ deterministically in the above construction, hence $\ex{\tilde{x}_p(i)} = \ex{x(i)}$. The claim follows since the adversary makes the optimal choices in order to minimize $\ex{\tilde{x}_p(i)}$.
\end{proof}
\begin{lemma}\label{lem:boundDeletions}
$\ex{\tilde{x}_p(i)} \geq \frac{\eps^4}{1024}2^ip$.
\end{lemma}
\begin{proof}
Consider the collection $\cS$ of $p$ sets and the deletion order chosen by the adversary. Once the $p$ pivots $\tilde{P}(i)$ are fixed, the best strategy for the adversary is to choose the sub-collection $\cS'$ of precisely $\frac{\eps^2}{8}p$ sets $S$ with smallest $\tilde{i}(S)$ (breaking ties arbitrarily). It remains to bound the expected value of $\tilde{x}_p(i)=\sum_{S\in \cS'}\tilde{i}(S)$. 

We say that a set $S\in \cS$ is \emph{bad} if $\tilde{i}(S)\leq \frac{\eps^2}{32} 2^{i}$ and \emph{good} otherwise. We let $b(i)$ and $g(i)$ be the number of bad and good sets, respectively. Observe that each set is bad independently with probability at most $\frac{\eps^2}{32}$, hence $\ex{b(i)}\leq \frac{\eps^2}{32}p$. By Markov's inequality, 
$$\pr{b(i)\geq \frac{\eps^2}{16}p}\leq \frac{1}{2}.$$ 
Given the event ${\mathcal{E}}=\left\{b(i)<\frac{\eps^2}{16}p\right\}$, one has that at least one half of the $\frac{\eps^2}{8}p$ selected sets are good, in which case deterministically 
$$\tilde{x}_p(i)\geq \frac{\eps^2}{16}p\cdot \frac{\eps^2}{32} 2^{i}=\frac{\eps^4}{512}2^ip.$$ 
We can conclude that
$$
\ex{\tilde{x}_p(i)}
\geq 
\pr{{\mathcal{E}}}\cdot \ex{\tilde{x}_p(i) \;\middle\vert\; {\mathcal{E}}}
\geq \frac{1}{2}\cdot \frac{\eps^4}{512}2^ip.\qedhere
$$
\end{proof}

Finally, we put the above lemmas together to obtain the desired bound.

\begin{lemma}\label{lem:decrementalTime}
The expected running time of the algorithm in the decremental case is $O\left(fn+\frac{f^2}{\eps^5}\cdot t\right)$, where $t$ is the number of deletions.
\end{lemma}
\begin{proof}
Let us consider a given level $i$. One has 
\begin{align}
\ex{t(i) \;\middle\vert\; i\in GC, p(i)=p} 
& \geq \ex{x(i) \;\middle\vert\; i\in GC, p(i)=p} \nonumber\\ \overset{\text{Lem. \ref{lem:domination}}}{\geq} \ex{\tilde{x}_p(i)}
& \overset{\text{Lem. \ref{lem:boundDeletions}}}{\geq} \frac{\eps^4}{1024}2^i p.\label{eqn:time1}
\end{align}
Hence
\begin{align}
& \ex{t(i) \;\middle\vert\; i\in GC} 
= \sum_{p}\pr{p(i)=p \;\middle\vert\; i\in GC} \cdot \ex{t(i) \;\middle\vert\; i\in GC, p(i)=p}\nonumber\\ \overset{\eqref{eqn:time1}}{\geq} & 
\sum_{p}\pr{p(i)=p \;\middle\vert\; i\in GC} \cdot \frac{\eps^4}{1024}2^i p 
= \frac{\eps^4}{1024}2^i \cdot \ex{p(i) \;\middle\vert\; i\in GC} \nonumber\\
\geq & \frac{\eps^4}{1024} 2^i \cdot \ex{d(i) \;\middle\vert\; i\in GC}\label{eqn:time2}.
\end{align}
Now, we note that 
\begin{align*}
& \sum_{i=0}^{L}\pr{i\in GC}\cdot \frac{f^2}{\eps} 2^i \ex{d(i) \;\middle\vert\; i\in GC} 
\\ \overset{\eqref{eqn:time2}}{\leq} & \sum_{i=0}^{L}\pr{i\in GC}\cdot \frac{f^2}{\eps} 2^i \frac{1024}{\eps^4 2^i} \ex{t(i) \;\middle\vert\; i\in GC}\\
\leq & \sum_{i=0}^{L}\frac{1024 f^2}{\eps^5}  \ex{t(i)}=\frac{1024 f^2}{\eps^5} t.
\end{align*}
The lemma now follows from Lemma~\ref{lem:timeGC}.
\end{proof}

We summarize the results in the following theorem.
\begin{theorem}
\label{theorem:main-decremental}
Given an $\eps >0$, let $\Delta=\frac{f^2}{\eps^5}$. There exists a decremental algorithm for set cover that achieves an $f(1+\epsilon)$ approximation and takes $O(\Delta \cdot t)$ time in expectation over $t$ updates.
\end{theorem}

\section{The Fully Dynamic Set Cover Algorithm}
\label{sec:fullyDynamic}
In this section, we extend the algorithm for the decremental case to the fully dynamic case.
At any time $t$, let $A \subseteq X$ denote the elements that need to be covered; we call these the {\em active} elements. 
Our goal is to maintain a feasible set cover $\cF$ for the active elements $A$ and ensure that the cost of $\cF$ is at most $(1+\epsilon)f$ times that of an optimal solution.
At the beginning, $A=\emptyset$ and $\cF=\emptyset$. Elements are then inserted or deleted from $A$ in a fixed sequence, independent of the randomness of the algorithm. If an element is inserted and then gets deleted and reinserted, we treat the two insertions separately as two copies of the same element.

\subsection{The Algorithm}

We now describe the {\em update phases} where the algorithm changes its solution in response to the insertions and deletions of elements. The update phases are very similar to the decremental algorithm, but with a few critical changes. To describe the changes, we need to introduce some additional notation. Just like the decremental algorithm, the fully dynamic algorithm maintains a set of pivots $P$, and at any time, the solution $\cF$ can be completely specified by $P$ as follows: $\cF=\{S \mid S \ni p\}$. $S(p)$ denotes the set of elements from which a pivot $p \in P$ is chosen and $\cX(p)$ denotes the set of elements $p$ is accounted to cover at any point of time. If $|S(p)| \in [2^i, 2^{i+1})$, we say that pivot $p$ is a level-$i$ pivot, and denote $\ell(p) = i$. We call the sets $\{S \mid S \ni p, \ell(p)=i\}$ level-$i$ sets. In addition, we partition $X(p)$ into two subsets $Orig(p)$ and $Extra(p)$, that is $X(p)=Orig(p)\cup Extra(p)$ and $Orig(p)\cap Extra(p)=\emptyset$. An element $e \in Orig(p)$ is called an {\em original} element and an element $e \in Extra(p)$ is called an {\em extra} element. $Orig(p)$ consists of all elements that $p$ is accounted to cover at the time when $p$ was chosen to be a pivot and $\cF^+= \{S\in \cS': p\in S\}$ sets are included in the solution.  Thus, $S(p) \subseteq Orig(p)$. It is possible that $p$ is accounted to cover more elements due to later updates. Those elements are added to $Extra(p)$. Along with $P$, the algorithm also maintains sets $D$ and $U$ of deleted and undeleted pivots respectively, just like in the decremental algorithm. When an element $e\in P$ is deleted, we move $e$ from $U$ to $D$ but keep this element in $P$. Changes to $P$ are done only during an update phase which we describe below. 

\smallskip\noindent{\bf Insertion of a new active element.} Supppose a new element $e$ is inserted in the set of active elements $A$. If $\{S \mid S \ni e\} \cap \cF \neq \emptyset$, then $e$ is already covered by the current solution $\cF$. In this case, let  $S$ be the set containing $e$ at the highest level breaking ties arbitrarily. If $p \in P \cap S$ denotes the pivot in $S$, then we insert $e$ in $Extra(p)$ and $X(p)$, and $\cF$ remains unchanged. Otherwise, $e$ is not covered by the current solution $\cF$. In this case, we include $e$ as a level-$0$ pivot and set $S(e)=\{e\}$, $X(e)=Orig(e)=\{e\}$. We update $\cF=\cF \cup \{S \mid S \ni e\}$. 
We also update the sets $P$ and $U$ to include $e$. 

\smallskip\noindent{\bf Deletion of an existing active element.}
When an element $e$ is deleted from the set of active elements $A$, we mark $e$ as deleted from the sets $\{S \mid S \ni e\} \cap \cF$. If $e\in P$, we move $e$ from $U$ to $D$. By doing so, if $|D| > \epsilon \cdot |P|$, then we say that an update phase has been triggered, and perform the following additional steps.

\smallskip\noindent{\bf Update Phase.}
First, we fix a critical level $\ell$ using the {\em level fixing process} of the decremental algorithm. Having fixed this level $\ell$,
we discard all sets $\cF(p)$ from $\cF$ that were added by pivots $p$ at levels $\ell$ or lower, i.e., where $\ell(p) \leq \ell$. Correspondingly, these pivots $p$ are also removed from $P$ and from either $D$ or $U$ depending on whether they are deleted or undeleted. As a result, a set of active elements become uncovered in $\cF$; this set is denoted $X'$. 

Next, the update phase has two steps, a {\em movement step} and a {\em covering step}, to cover the elements in $X'$.

\noindent{\bf Movement step:} For each element $e \in X'$, we check if there is a set $S \in \cF$ containing $e$ at a level $\ell' > \ell$. If yes, we select a set $S \ni e$ at the highest level (breaking ties arbitrarily). If $p=P \cap S$, then $e$ is added to $Extra(p)$ and $X(p)$.  

\noindent{\bf Covering step:} Let $Y' \subseteq X'$ denote the elements left uncovered after the movement step. We now run the Random Cover subroutine on the instance $(Y', \cS')$ induced by $Y'$ and add the resulting sets $\cF'$ to the overall solution $\cF$. We say that levels $\ell$ and below have been updated in the current update phase. The newly selected pivots are added to $P$ and $U$. For every newly chosen pivot $p \in Y' \cap P$, if $S(p) \in [2^i, 2^{i+1})$, then pivot $p$ is a level-$i$ pivot and we include $\{S \mid S \ni p\}$ at level $i$. Note that it is possible that $i > \ell$ due to newly inserted elements. Also note that all the elements of $Y'$ now become original elements after the covering step. 

The next lemma, shown in Section~\ref{sec:running-time}, establishes the time taken to implement the above algorithm.
\begin{lemma}
\label{lma:update-dynamic}
The above algorithm for the insertion of a new element, or the deletion of an existing element that does not trigger an update phase, takes $O(f)$ time. The time complexity of an update phase is $O(f|X'|)$.
\end{lemma}

\subsection{Analysis of the Competitive Ratio}

\begin{lemma}
\label{lem:apxFull}
The competitive ratio of the algorithm is at most $f/(1-\eps)$.
\end{lemma}
\begin{proof}
Consider the data structure right before the $t$-th update. Let $P^t$ be the set of pivots at that time, with $U^t$ being the subset of undeleted pivots. We also let $OPT^t$ and $\cF^t$ be the optimal and approximate solution at that time. 

Observe that $|\cF^t|\leq f\cdot |P^t|$ by construction. We claim that $|OPT^t|\geq |U^t|$. This implies the claim since by construction $|P^t|\leq |U^t|/(1-\eps)$ at any time.

To see that, let us show by a simple induction that, for any two distinct $p,p'\in U^t$, there is no set $S\in \cF$ covering both $p$ and $p'$. Thus $OPT^t$ needs to include a distinct set for each element of $U^t$. Since we start with an empty set cover, the property holds at the outset. Assume the property holds up to step $t\geq 1$, and consider step $t+1$. If an element $e$ is inserted at time $t+1$, then $e$ becomes a pivot if and only if $e$ is not covered by any existing set in $\cF$. Hence, the property holds. If a non-pivot element $e$ is deleted at $t+1$, or $e$ is a pivot but its deletion does not trigger an update phase, then since we do not change the solution $\cF$, the property holds by the inductive hypothesis.

Now, assume $e$ is a pivot and its deletion triggers an update phase at critical level $\ell'$, with the elements covered at levels $i\leq \ell$ being denoted by $X'$. Note that we select a set of new pivots from $Y' \subseteq X'$ and let $\cF'$ denote the new sets that are added after the update phase. We consider three cases. if $p, p'\in \cF$, then they do not belong to the same set by the inductive hypothesis. If $p, p'\in \cF'$, then they do not belong to the same set since the Random Cover subroutine picked both these elements as pivots. Finally, if $p\in \cF$ and $p'\in \cF'$, then the movement step ensures that $p'$ is not covered by $\cF$ whereas all sets containing $p$ are in $\cF$. Therefore, $p$ and $p'$ do not belong to the same set in this case either. Therefore, the property holds after the $(t+1)$-st  update.
\end{proof}

\subsection{Analysis of the Amortized Update Time}

Consider a sequence of $t$ updates, and let $T$ be the time right before the $(t+1)$-st update occurs (or the end of the execution if no such update occurs). Our goal is to bound the 
expected update time till $T$ as $O\left(t\cdot \frac{f^2\log{n}}{\epsilon}\right)$.
We recall some definitions from Section~\ref{sec:decremental} and introduce some new notation for the purpose of the analysis. 

\noindent{\bf Old notation.} Recall the definition of the critical level $\ell$ and Lemma~\ref{lem:matching}. Note that when performing an update at a critical level $\ell$, a level $i \leq l$ is said to be charged if $d_i \geq \frac{\epsilon}{2} p_i$ and uncharged otherwise.



Suppose that level $i$ is updated $q_i$ times in total, and let $P_i^j$ be the total set of pivots in the $j$-th such update. $P(i)$ is the multiset of pivots obtained by taking union over $P_i^j$ and $p(i)=|P(i)|$. The pivots of type $P_i^j$ where level $i$ is charged on the $j$-th update are called charged and denoted $CH(i)$, $ch(i)=|CH(i)|$. The charged deleted pivots at level $i$ are denoted $D(i), d(i)=|D(i)|$.  Let $I$ be the total number of insertions up to time $T$.

Recall that a level $i$ is globally-charged iff $ch(i) \geq \frac{\epsilon}{4} p(i)$, and globally-uncharged otherwise. Let $GC$ denote the random set of globally-charged levels.

\noindent{\bf New notation.} Let us now define a new mapping $R$ that maps an element $e$ (on which an update phase operates) either to a charged deleted pivot or to an insertion. To construct this mapping, if an element $e$ takes part in $q_e$ update phases, including both the movement and the covering steps, then each of these occurrences is treated separately. 

An element $e \in Orig(p)$ is mapped to a charged deleted pivot $d$, i.e., $R(e)=d$, if $M(p)=d$. Now consider an element $e \in Extra(p)$. If $e$ has never been an original element, then it must have been inserted as an extra element and has taken part only in movement steps since then. This is because whenever a covering step processes an element, it becomes an original element. In this case, we map $e$ to its insertion, denoted $e_I$ and set $R(e)=e_I$. Otherwise, consider the last update phase, when $e$ was an original element just before the update and became an extra element immediately after it. If $e \in Orig(p')$ and $M(p')= d'$ during that phase, then set $R(e)=d'$.

Note that all $q_e$ occurrences of $e$ are mapped by $R$ to either to the insertion $e_I$ or to charged deleted pivots. If $R(e)=e_I$, we say insertion $e_I$ is responsible for $e$ and if $R(e)=d$, we say the charged deleted pivot $d$ is responsible for $e$. Note that a charged deleted pivot $d$ is responsible for an element $e$ if an only if $e \in Orig(p)$ for some pivot $p$ and $M(p)=d$.

Let us use $L=\lfloor \log_2 n\rfloor$ to denote the largest level and $\mathcal{X}$ to denote the multiset of elements obtained by taking the union of $Orig(p)$ and $Extra(p)$ over all $p \in P$.

\begin{lemma}
\label{lem:responsible-1}
    Each charged deleted pivot $d\in D(i)$ is responsible for at most  $(2/\epsilon)\cdot f\cdot 2^{i+1}\cdot (L+1)$ elements of $\mathcal{X}$.
\end{lemma}
\begin{proof}
Consider a deleted pivot $d \in D(i)$ and let $P_d=\{p \mid M(p)=d\}$ be the set of pivots mapped to $d$. Consider any pivot $p \in P_d$. Note that, $d$ is only responsible for the elements in $Orig(p)$ and that $\ell(p)\leq \ell(d)=i$. By definition of $\ell(p)$, each set containing $p$ covers less than $2^{\ell(p)+1}$ new elements at the time $p$ was selected to be a pivot. Hence $|Orig(p)|< f\cdot 2^{\ell(p)+1}\leq f\cdot 2^{i+1}$. Now let us count the number of times $d$ was held responsible for $e \in Orig(p)$. The element $e$ was an original element just before the update phase that operates on $d$. If $e$ gets processed during that phase by the covering step, then $d$ was responsible for $e$ only once. Otherwise, $e$ gets processed by the movement step during that phase and becomes an extra element. If $e$ is processed $r$ times by the movement step before becoming an original element again, then $d$ is responsible $r+1$ times for $e$. However, the level of $e$ strictly increases after each movement step. Therefore, $r \leq L$.
Thus, $d$ can be responsible for $e$ at most $L+1$ times. This holds for all $p \in P_d$. Now, the claim follows noting $|P_d|\leq 2/\eps$ by Lemma~\ref{lem:matching}.   
\end{proof}

\begin{lemma}
\label{lem:responsible-2}
    Each insertion $e_I \in I$ is responsible for at most $L+1$ elements of $\mathcal{X}$.
\end{lemma}
\begin{proof}
An insertion $e_I$ is only responsible for the element $e$. If $e$ takes part in $r$ movement steps before becoming an original element for the first time, then $e_I$ is responsible $r+1$ times for $e$.  Since the level of $e$ strictly increases after each movement step, we have $r \leq L$. Thus, the claim follows.
\end{proof}

We now have an analog of Lemma~\ref{lem:timeGC}. 
\begin{lemma}\label{lem:timeGC-dyn}
The expected running time of the algorithm after $t$ updates is
$$
O\left(f\cdot |I|\cdot (L+1)+\sum_{i=0}^{L} \pr{i\in GC}\cdot \ex{\frac{f^2}{\eps}\cdot 2^i\cdot  d(i)\cdot (L+1) \;\middle\vert\; i\in GC}\right).
$$
\end{lemma}
\begin{proof}
From Lemma~\ref{lma:update-dynamic}, the total update time up to time $T$ is $O(f\cdot (|I|+|\mathcal{X}|))$. Now from Lemma~\ref{lem:responsible-1} and Lemma~\ref{lem:responsible-2}, 
\begin{equation*}
    f\cdot (|I|+|\mathcal{X}|)
    = f\cdot \left((L+1)\cdot |I| + (2/\epsilon) \cdot f\cdot  (L+1)\cdot \sum_{i}d(i)\cdot 2^{i+1}\right).
\end{equation*}
We can give each charged deleted pivot $\frac{4}{\epsilon}\cdot f\cdot (L+1)\cdot 2^{i}$ tokens and then follow the token transfer process of Lemma~\ref{lem:timeGC} so that all these tokens are transferred to charged deleted pivots in the globally charged levels. Moreover, each charged deleted pivot in a globally charged level contains at most $\frac{8}{\epsilon}\cdot f\cdot (L+1)\cdot 2^{i}$ tokens. The lemma now follows.
\end{proof}

Let $t(i)$ denote the (random) number of deletions at level $i$ up to time $T$, and let $t'=\sum_{i}t(i)$. Then $t' \leq t$. 
Exactly as in the decremental case, we can now use Lemmas~\ref{lem:domination},~\ref{lem:boundDeletions}, and~\ref{lem:decrementalTime} to bound
\begin{equation}
\label{eq:bound-dynamic}
\sum_{i=0}^{L} \pr{i\in GC}\cdot \ex{\frac{f^2}{\eps}2^i d(i)(L+1) \;\middle\vert\; i\in GC}
\leq \frac{1024}{\epsilon^5} f^2 t (L+1).
\end{equation}

%
\begin{lemma}\label{lem:fullydynamicTime}
The expected running time of the algorithm in the fully dynamic case is $O\left(\frac{f^2 \log{n}}{\eps^5}\cdot t\right)$, where $t$ is the number of updates.
\end{lemma}
\begin{proof}
This follows from Lemmas~\ref{lem:decrementalTime} and ~\ref{lem:timeGC-dyn}, and Eq.~\eqref{eq:bound-dynamic}, 
noting that $t' \leq t$, $|I| \leq t$, and $L =\O(\log{n})$. 
\end{proof}

We summarize the results in the following theorem.
\begin{theorem}
\label{theorem:main-dynamic}
Given an $\eps >0$, let $\Delta=\frac{f^2\log{n}}{\eps^5}$. There exists a fully-dynamic algorithm for set cover that achieves an $f(1+\epsilon)$ approximation and takes $O(\Delta \cdot t)$ time in expectation over $t$ updates.
\end{theorem}

\section{Implementation Details and Running Time}
\label{sec:running-time}
In this section, we give implementation details of the algorithms, leading to the proofs of Lemma~\ref{lma:update-runtime} (decremental) and Lemma~\ref{lma:update-dynamic} (fully dynamic). 

\subsection{Decremental Algorithm}

We assume that any given set cover instance $(X',\cS')$, with maximum frequency $f$ is represented as follows. Elements (resp., sets) are labelled $1$ to $n'=|X'|$ (resp., $m'=|\cS'|$). W.l.o.g. we can assume that each set covers at least one element of $X'$, so that $m'\leq f\cdot n'$. We have a vector $SET$ indexed by elements, where $SET[e]$ is the list of sets $\cS'(e)$ containing $e$. Observe that $SET[e]$ contains at most $f$ entries. We assume that sets are described by a vector $ELEM$ indexed by sets, where $ELEM[S]$ is a list of elements contained in set $S$. We keep a pointer from each $e\in ELEM[S]$ to the corresponding entry $S$ in $SET[e]$ and vice-versa. 

In order to implement deletions, we proceed as follows. We maintain a Boolean vector $DEL$ indexed by $e\in X'$, which is initialized to false. When element $e$ is deleted, we set $DEL[e]=true$. Furthermore, we scan $SET[e]$, and for each $S\in SET[e]$ we remove $e$ from $ELEM[S]$. Note that this can be done in $O(1)$ time for each set $S$ using the pointers mentioned above, i.e., in time $O(f)$ per element $e$. This also implies that deleting all the elements one by one takes $O(f|X'|)$ time in total.

\noindent
{\bf The Random Cover Subroutine.} 
The Random Cover procedure computes a sequence of pivots $P'$. Furthermore, for each pivot $p\in P'$, it computes a collection $\cF(p)$ of sets that are added to the solution because of $p$, and the corresponding set $\cX(p)$ of newly covered elements. This takes $O(f+|\cX(p)|)$ time for a given pivot $p$.


It remains to specify how we efficiently extract a set $S$ of maximum cardinality at each step to select a pivot. We maintain a list $SORT$ whose entries are pairs $(i,L_i)$, where $i$ is the cardinality of a set and $L_i$ is the list of sets of cardinality $i$. We store all such entries with $L_i$ not empty, in decreasing order of $i$. This list can be initialized in linear time $O(f|X'|)$ (say, using radix sort). We also maintain pointers from each set $S$ to the corresponding entry in the list $L_{|S|}$.

The first element of the first list $L_i$ is the selected set $S$ at each step. Then we update $SORT$ as follows. Each time we remove an element $e$ from some set $S'$ of cardinality $i$, we remove $S'$ from $L_{i}$ and add it to $L_{i-1}$. Note that this might involve creating a new entry $(i-1,L_{i-1})$ in $SORT$ (if $S$ becomes the only set of cardinality $i-1$), or deleting the entry $(i,L_i)$ from $SORT$ (if $S$ was the only set of cardinality $i$). In any case, these operations can be performed in $O(1)$ time. It follows that the entire procedure can be implemented in time $O(f\cdot |X'|)$ time.

\noindent
{\bf The Set Cover Solution.} We store and maintain the approximate solution as follows. We maintain the set cover instance under deletions as described before. Furthermore, we maintain vectors $\cF$ and $\cX$. For a pivot $p$, $\cF(p)$ is the corresponding list of selected sets because of $p$, and $\cX(p)$ is the associated list of newly covered elements due to these selected sets. These two lists are empty if $p$ is not a pivot.  

\noindent{\bf Level selection.} We maintain counters $D$ and $P$ labelled by levels $i=0,\ldots,\lceil \log_2 n\rceil$, where $D[i]$ (resp., $P[i]$) is the number of deleted pivots (resp., all pivots) at level $i$. When we delete a pivot at level $i$, we increment $D[i]$. When we update at critical level $\ell$, we set $D[i]=P[i]=0$ for all $i\leq \ell$. Furthermore we increment $P[i]$ for each newly computed pivot of level $i$. Clearly these operations have amortized cost $O(1)$ per update. We similarly maintain the total number $\tilde{D}$ and $\tilde{P}$ of deleted pivots and all pivots, respectively. By comparing $\tilde{D}$ and $\tilde{P}$ at each deletion, we can check whether the condition for the update of a suffix is satisfied. In that case, using $D$ and $P$, it is easy to compute in $O(\ell^2)$ time the lowest critical level $\ell$.

\noindent
{\bf Update Phase.} We next describe how, given a critical level $\ell$, we update the approximate solution. We keep a list $GREEDY$ whose entries are pairs $(j,P_j)$. Here $j$ is a level and $P_j$ is the list of pivots of that level. We keep such entries only for non-empty $P_j$, in increasing order of $j$.

Given a critical level $\ell$, we scan the list $GREEDY$ and compute $P':=\cup_{j\leq \ell}P_j$ together with $X':=\cup_{p\in P'}\cX(p)$ (represented as lists). We remove all the corresponding entries from $GREEDY$, and reset the corresponding values of $\ell(p)$, $S(p)$, $\cF(p)$, and $\cX(p)$.

Let us show how to build the data structures for the subinstance $(X',\cS')$, with $\cS'=\{S\cap X': S\in \cS\text{ and } S\cap X'\neq \emptyset\}$, in time $O(f|X'|)$. Let $n'=|X'|$ and $m'=|\cF'|\leq f n'$. By scanning the entries of $SET$ corresponding to $X'$, we build the list of indexes $\cS'$. We now map $X'$ into a set of new indexes in $[1,O(n')]$ by means of a perfect hash function, and similarly map $\cF'$ into a set of new indexes in $[1,O(m')]$. These perfect hash functions can be built in expected linear time using well-known constructions (say, $2$-level hashing \cite{FKS84}). Observe that some indexes might not be used: we interpret those indexes as dummy elements and sets.
Given this, we can easily build in linear time the data structures $SET'$ and $ELEM'$ for the new instance.  

The use of random hash functions can be avoided by assuming that we are given access to two arrays $\elemm$ and $\setm$ respectively of size $n$ and $m \leq nf$ that are initialized to all zeros. We now map $X'$ to $[1,n']$ and $\cS'$ to $[1,m']$ as follows. We create vectors $\elemm^{-1}$ and $\setm^{-1}$ of size $n'$ and $fn'$, resp., which are initialized to zero. We iterate over $X'$: when considering the $j$-th element of $X$ of global id $k$ (that is, it is the $k$-th element in $[1,n]$), then we set $\elemm[k]=j$ and $\elemm^{-1}[j]=k$. Similarly, we iterate over $\cS'$ and update $\setm$ and $\setm^{-1}$ analogously. In order to handle possible duplicates in $\cS'$ (and possibly in $X'$), we simply do not perform the update when we find an entry of $\setm$ (or $\elemm$) that is already non-zero.
Now, $X'$ has been mapped to $[1,n']$ and $\cS'$ has been mapped to $[1,m']$. This allows us to build vectors $ELEM'$ and $SET'$ in the same way as above.
Once the update phase has ended, we reset $\elemm$ and $\setm$ to $0$ by iterating over $\elemm^{-1}$ and $\setm^{-1}$. The overall process takes $O(fn')$ time.

We feed the vectors $SET'$ and $ELEM'$ to the Random Cover subroutine that outputs a list $P'$ of pivots, plus the associated values $\ell(p)$, $S(p)$, $\cF(p)$, and $\cX(p)$ for each $p\in P'$. Using $\elemm^{-1}$ and $\setm^{-1}$ we can map back the indexes of the corresponding elements and sets into the original indexes.

We remark that by construction, we will have $\ell(p)\leq \ell$. Now, we build a list $GREEDY'$ of the same type as $GREEDY$, however restricted to pivots in $P'$ and to the respective levels. Finally, we concatenate $GREEDY'$ to the beginning of the list $GREEDY$.

We are now ready to prove Lemma \ref{lma:update-runtime}.
We first observe that, up to constant factors, the time taken in one level fixing process at critical level $\ell$ is at most that in the subsequent update phase. Indeed, recall that level fixing at level $\ell$ takes time $O(\ell^2)$. Since $\ell$ is the lowest critical level,
$\ell-1$ is not critical, therefore there is at least one deleted pivot in 
level $\ell$. This implies that this update involves at least $2^\ell$ elements, thus having cost $\Omega(2^{\ell})$. The claim follows. 

We can therefore focus on the cost of the update phase. Each pivot $p$ at level $i$ that participates in the update corresponds to at most $f$ sets of size at most $2^{i+1}$ each. Hence the set $X'$ of elements that participate in the update has size at most $\sum_{i\leq \ell}fp_i 2^{i+1}$. As argued above, we can build the corresponding set cover instance $(X',\cS')$ and run the Random Cover subroutine on it in $O(f|X'|)$ time. This completes the proof of Lemma~\ref{lma:update-runtime}.

\subsection{Fully Dynamic Algorithm}
We now briefly describe the changes in the fully-dynamic algorithm that leads to Lemma~\ref{lma:update-dynamic}.
Along with each $\mathcal{X}(p)$, we also maintain two disjoint subsets of $\mathcal{X}(p)$: $Orig(p)$ and $Extra(p)$. When a set $S$ is included in the current solution with pivot $p$ and newly covers elements $\mathcal{X}(p)$, we follow the implementation details of the decremental case. In addition, we set $Orig(p)=\mathcal{X}(p)$ and $Extra(p)=\emptyset$. This does not change the asymptotic run time.

When an element $e$ is inserted, we iterate over $SET(e)$ and check if there exists any set $S \in SET(e)$ present in the current solution. For each $S \in SET(e)$, the check can be implemented in $O(1)$ time by maintaining a Boolean variable for every $S \in \mathcal{S}$ and setting it to $1$ whenever it is included in the solution. If $S \in \mathcal{S}$, we also maintain a pointer to its pivot $p$, and the value $\ell(p)$. If none of the sets in $SET(e)$ is present, then we include all sets in $SET(e)$ in the solution, mark $e$ as the pivot for these sets and set $\ell(e)=0$. We set $\mathcal{X}(e)=\{e\}$, $Orig(e)=\{e\}$ and $Extra(e)=\emptyset$. This entire process can be implemented in $O(f)$ time as $|SET(e)| \leq f$. Thus, insertion takes $O(f)$ time. 

Deletion without an update phase is implemented in the same way as in the decremental case and, therefore, takes $O(f)$ time as well.
If there is an update phase at a critical level $\ell$, given the instance $(X',\mathcal{S}')$, we first execute the movement step as follows. We scan every element $e \in X'$ and iterate over $SET(e)$ to check if there exists a set $S \in SET(e)$ that is in the current solution at a level strictly higher than $\ell$. If so, we include $e$ as an extra element in the highest such set and discard $e$ from $X'$. This entire operation takes $O(f|X'|)$ time. 

Let $Y'$ be the uncovered elements at the end of the movement step. We run the covering step as in the decremental case over $Y'$. This takes $O(f|Y'|)$ time which is bounded by $O(f|X'|)$. Whenever a new pivot $p$ is selected, the level of $p$ and the corresponding sets are determined in $O(|\mathcal{S}(p)|)$ time. Accordingly, the lists $L_i$s are updated. These updates take $O(1)$ time per set. 

Thus, the update phase in the fully dynamic algorithm can be implemented in $O(f|X'|)$ time, thereby proving Lemma~\ref{lma:update-dynamic}.
 

\section{Conditional Lower Bounds for Dynamic Set Cover}
\label{sec:lb}
A fast algorithm for a dynamic problem usually gives a fast algorithm for its static version. If we can solve dynamic Set Cover with preprocessing time $P(n,f)$ and update time $T(n,f)$, then we can solve the static Set Cover problem in $P(n,f) + n \cdot T(n,f)$ time: $n$ updates are sufficient in order to create an offline instance. This simple connection immediately leads to some lower bounds for dynamic Set Cover. In particular, we get that it is NP-Hard to get an $o(\log{n})$ approximation with polynomial preprocessing and update times. However, this connection does not give any lower bound in the polynomial time solvable regime of Set Cover where a $\min \{f, \log n \}$ approximation is possible in linear time. 
Could there be a dynamic algorithm with such approximation factors that has $o(f)$ or even $O(1)$ update time? This would not imply a new static algorithm for Set Cover.

Of course, such an algorithm is impossible if to insert an element we must explicitly specify the $O(f)$ sets it appears in. But in the model we consider, an update can be specified with much fewer bits. 
We assume that all elements $X$ and sets $\mathcal{S}$ are given in advance, as well as all the membership information. Then, an update can add or remove an elements (in the Element-Update case) or sets (in the Set-Updates case). Only elements from $X$ (or sets from $\mathcal{S}$) can be added or removed, and when an element is removed then a set-cover does not need to cover it. Notice that $O(\log{|X|})$ bits are needed to specify an element insertion, and its membership in all the sets is known from the initial input.
In this model, it is conceivable that an algorithm can spend $o(f)$ time per update and maintain some non-trivial approximation. 
The results in this section show that this is unlikely.

Under SETH, we show that no algorithm can preprocess an instance with $m$ sets and $n$ elements in $poly(n,m)$ time, and subsequently maintain element (or set) updates in $O(m^{1-\eps})$ time, for any $\eps>0$, unless the approximation factor is essentially $m^{\delta}$, for some $\delta>0$. Note that a factor $m$ approximation can be maintained trivially in constant time (pick either zero or all sets). and we show that essentially any $m^{o(1)}$ approximation algorithm requires $\Omega(m^{0.99})$ time.

\begin{theorem}[Main Lower Bound]
\label{thm:lb}
Let $n^{\omega(1)} < m < 2^{o(n)} $ and $t \geq 2$, such that $t = (n/ \log{m})^{o(1)}$.
Assuming SETH, for all $\eps>0$, no dynamic algorithm can preprocess a collection of $m$ sets over a universe $[n]$ in $poly(n,m)$ time, and then support element (or set) updates in $O(m^{1-\eps})$ amortized time, and answer Set Cover queries in $O(m^{1-\eps})$ amortized time with an approximation factor of $t$.
\end{theorem}

We can state the following corollary in terms of the frequency bound $f$.

\begin{corollary}
Assuming SETH, any dynamic algorithm for Set Cover on $n$ elements and frequency bound $f$, where $n^{\omega(1)} < f < 2^{o(n)}$, that has polynomial time preprocessing and amortized update and query time $O(f^{1-\eps})$, for some $\eps>0$, must have approximation factor at least $(n/\log{f})^{\Omega(1)}$.
\end{corollary}

\begin{proof}
Without assuming anything about the instances in Theorem~\ref{thm:lb} we can conclude that $f \leq m$ while $m \leq n f < f^2$. 
Therefore, any approximation algorithm with factor $O((n/\log{f})^{\delta})$ also gets an approximation of $O((n/2\log{f})^{\delta}) = O((n/\log{f^2})^{\delta})$ which is smaller than $O((n/\log{m})^{\delta})$ and it is enough to refute SETH via Theorem~\ref{thm:lb}. 
\end{proof}

The rest of this section is dedicated to the proof of Theorem~\ref{thm:lb}. 
Our starting point is the following SETH-based hardness of approximation result, which was proven first in \cite{ARW17} with a slightly smaller approximation factor, and was strengthened in \cite{Chen18} using, in part, the technique of \cite{Rub18}.
These results use the distributed PCP framework of \cite{ARW17} for hardness of approximation results in P, and ours is the first application of this framework to dynamic problems.

\begin{theorem}[\cite{ARW17,Rub18,Chen18}]
\label{thm:ARW}
Let $n^{\omega(1)} < m < 2^{o(n)} $ and $t \geq 2$, such that $t = (n/ \log{m})^{o(1)}$.
Given two collections of $m$ sets $\mathcal{A},\mathcal{B}$ over a universe $[n]$, no algorithm can distinguish the following two cases in $O(m^{2-\eps})$ time, for any $\eps>0$, unless SETH is false:

\begin{description}
\item [{YES case}] there exist $A \in \mathcal{A}, B \in \mathcal{B}$ such that $B \subseteq A$;
and
\item [{NO case}] for every $A \in \mathcal{A}, B \in \mathcal{B}$ we have $| A \cap B | <  |B|/t$.
\end{description}

\end{theorem}

From this theorem and standard manipulations it is easy to conclude the following statement.
There are two differences in the statement below: first, the sizes of $\mathcal{A}$ and $\mathcal{B}$ are asymmetric, and second, the approximation is in terms of the number of sets required to cover a single $b \in B$, rather than the size of the overlap.

\begin{lemma}
\label{lem:BCP}
Let $n^{\omega(1)} < m < 2^{o(n)} $ and $t \geq 2$, such that $t = (n/ \log{m})^{o(1)}$, and for all $0<a \leq 1$.
Given two collections of sets $\mathcal{A},\mathcal{B}$ over a universe $[n]$, where $|\mathcal{B}|=m$ and $|\mathcal{A}|=m^a$, no algorithm can distinguish the following two cases in $O(m^{1+a-\eps})$ time, for any $\eps>0$, unless SETH is false:

\begin{description}
\item [{YES case}] there exist $A \in \mathcal{A}, B \in \mathcal{B}$ such that $B \subseteq A$;
and
\item [{NO case}] there do not exist $t$ sets $A_1,\ldots, A_t \in \mathcal{A}$, and a set $B \in \mathcal{B}$ such that $B \subseteq A_1 \cup \cdots \cup A_t$.
\end{description}

\end{lemma}

\begin{proof}
Assume for contradiction that such an algorithm exists.
Given an instance $\mathcal{A},\mathcal{B}$ of the problem in Theorem~\ref{thm:ARW} we show how to solve it in $O(m^{2-\eps})$ time.
Partition $\mathcal{A}$ into $k=m^{1-a}$ collections $\mathcal{A}_1,\ldots, \mathcal{A}_k$ of size $m^a$ each, and invoke our algorithm on the asymmetric instance $\mathcal{A}_i$, $\mathcal{B}$ for each $i =1 \cdots k$.
The total time will be $k \cdot O(m^{1+a-\eps})= O(m^{2-\eps})$.
If the original (symmetric) instance was a YES case, then clearly at least one of the $k$ asymmetric instances is a YES case.
On the other hand, if it was a NO case, then any $A \in \mathcal{A}$ cannot cover more than a $1/t$ fraction of any set $B \in \mathcal{B}$ and therefore all the asymmetric instances are NO cases.
 \end{proof}

Next we take this static set-containment problem and reduce it to dynamic Set Cover. We show two distinct reductions, a simpler one for the element updates case, and then a more complicated one with set updates.

\subsection{Element Updates}
Given an instance $\mathcal{A},\mathcal{B}$ of the problem in Lemma~\ref{lem:BCP}, we construct an instance of dynamic Set Cover with approximation factor $(t-1)$ as follows.
The universe $[n]$ will be the same, and all sets in $\mathcal{A}$ will appear in the instance.
However, the sets in $\mathcal{B}$ will not, and they will be implemented implicitly in a dynamic way.
Initially, all the universe elements are activated, and the algorithm may preprocess the instance. 
Note that the number of sets is only $m^a$.

For each set $B_i \in \mathcal{B}$ we will have a stage.
We start the stage by removing from the universe all elements $e \in B_i$ that belong to $B_i$.
After we do these $O(n)$ updates, we ask a Set Cover query.
If the answer is less than $t$ then we can stop and answer YES.
Otherwise, we finish the stage by adding back all the elements that we removed and move on to the next stage.
After we finish all $m$ stages for all the sets in $\mathcal{B}$, we answer NO.

In total we have $O(nm)$ updates and queries, and so the final runtime is $P(n,m^a)+O(nm) \cdot ( T(n,m^a)+Q(n,m^a))$. 
Assume we have an algorithm with update and query time $T(n,m^a)+Q(n,m^a)=O(m^{a \cdot (1-\eps)})$ and polynomial preprocessing, $P(n,m^a) =O(m^{a \cdot c})$ for some $c \geq 1$, then we can choose $a=1/c$ and get an algorithm for the problem in Lemma~\ref{lem:BCP} with runtime $O(m^{1+a-\eps a})$, contradicting SETH.

Finally, let us show the correctness of the answer.
If we are in the YES case, then there is a set $B \in \mathcal{B}$ that is contained in some set in $\mathcal{A}$. When we ask a query at the stage corresponding to this set $B$, the size of the minimum set cover is $1$. To see this note that all active universe elements are the elements of $B$ and so we can cover all of them with some set in $\mathcal{A}$.
Therefore, our $(t-1)$ approximation algorithm must output an answer that is less than $t$ and we will output YES.
On the other hand, if we are in the NO case, then in all stages, the size of the minimum set cover is at least $t$ since at least $t$ sets from $\mathcal{A}$ are required to cover any set in $\mathcal{B}$. Thus, the approximation algorithm will always return an answer that is at least $t$ and we will never output YES.
\subsection{Set Updates}
The previous reduction fails in this case because we are only allowed to update sets, not elements.
A natural approach for extending it is to have all sets from $\mathcal{B}$ in our instance and then at each stage we activate one of them.
This would work, except that the number of sets grows to $m$ which would only give us a weaker lower bound. 
Indeed such a simple reduction can rule out $O(m^{1-\eps})$ update times if the preprocessing is restricted to take subquadratic time.
A different idea is to add $n$ auxiliary sets, one per element, so that this set only contains that element. Then, if we want to remove an element, we can add this set and somehow ensure that it is a part of the solution so that, effectively, the corresponding element is removed.
This is the approach we take.
The main challenge, however, is that these auxiliary sets have to be picked in our set cover solution and so they contribute to the size of the optimal solution.
That is, we will no longer have a set cover of size $1$ in the YES case and the gap between the YES and NO cases changes. To overcome this issue, we introduce another idea where we create many copies of everything and combine them into one instance in a certain way. 
  
 Given an instance $\mathcal{A},\mathcal{B}$ of the problem in Lemma~\ref{lem:BCP}, we construct an instance of dynamic Set Cover with approximation factor $(t-1)$ as follows.
 
Our universe will be $k:=n^2$ times larger, and for each element $e \in [n]$ in the universe of the original instance we will add $k$ elements $e^1,\ldots,e^k$ to our instance.
(So, our universe is isomorphic to $[kn]$.)

For each set $A \in \mathcal{A}$ we construct $t$ sets $A^1,\ldots,A^k$ in our dynamic instance.
All of these sets will remain activated throughout the reduction.
The set $A^i$ contains all elements $e^i$ such that $e\in A$. That is, $A^i$ contains the $i^{th}$ copy of all the elements that were in $A$.
Note that $A^i$ does not contain $e^j$ for any $i \neq j$.

We also add sets $S_1,\ldots,S_n$ which will be activated dynamically, and we let $S_e$ contain all copies of the element $e \in [n]$. That is, $S_e$ contains $e^1,\ldots,e^k$.
These sets will allow us to simulate the deactivation of a set $B$.

Next we explain the dynamic part of the reduction.
For each set $B \in \mathcal{B}$ we have a stage where we effectively deactivate all universe elements that are not in $B$. 
To do this, we activate the set $S_e$ for all $e \notin B$ such that $e$ is not in $B$.
Note that we have activated up to $n$ sets $S_e$, and that together they cover all copies of all elements that are in the complement of $B$.
After we perform these $O(n)$ updates, we ask a Set Cover query.
If the answer to the query is at most $(n+k) \cdot (t-1)$ we return YES.
Otherwise, we undo the changes we made in this stage and we move on to the next $B \in \mathcal{B}$.
After all the stages are done, we return NO.

The runtime analysis is similar to before since the only difference is in the universe size which increased from $n$ to $kn=n^3$ but it is still $m^{o(1)}$.
We have $O(nm)$ updates and queries, and so the final runtime is $P(nk,m^a)+O(nm) \cdot (T(nk,m^a)+Q(nk,m^a))$. 
Assume we have an algorithm with update and query time $T(nk,m^a)+Q(nk,m^a)=O(m^{a \cdot (1-\eps)})$ and polynomial preprocessing, $P(nk,m^a) =O(m^{a \cdot c})$ for some $c \geq 1$, then we can choose $a=1/c$ and get an algorithm for the problem in Lemma~\ref{lem:BCP} with runtime $O(m^{1+a-\eps a})$, contradicting SETH.

Finally, we show the correctness of the answer.
For the YES case, there is a set $B \in \mathcal{B}$ that is contained in some set in $\mathcal{A}$. When we ask a query at the stage corresponding to this set $B$, the size of the minimum set cover is at most $n + k$.
This is because of the following set cover: Choose all sets $S_e$ that are active in this stage; this cover all copies of all universe elements that are not in $B$. Then choose all copies $A^i$ of the set $A \in \mathcal{A}$ that contains $B$; this covers all copies of all elements that are in $B$.
Therefore, our $(t-1)$ approximation algorithm must output an answer that is at most $(n+k)(t-1)$ and we will output YES.
On the other hand, in the NO case, the size of the minimum set cover is at least $k \cdot t$ in every stage.
This is because at least $t$ sets from $\mathcal{A}$ are required to cover any set in $\mathcal{B}$, and in a stage of some set $B$ the only way we can cover copies of elements that belong to $B$ is by choosing copies of sets $A$ that contain them. There are $k$ copies of the universe elements, and for each such copy we have to choose at least $t$ sets from $A$ to cover the elements of that copy, and these sets do not contain any elements from any other copy of the universe. 
Thus, the approximation algorithm will always return an answer that is at least $kt$, which is larger than $(n+k)(t-1)$ since $k=n^2$ and $t =n^{o(1)}$, and we will never output YES.

\newpage
\smallskip\noindent{\bf Acknowledgements.} The authors are grateful to Shay Solomon for pointing out an error in a preliminary version of this paper. The authors would also like to thank the anonymous reviewers for their insightful comments. 
\bibliographystyle{alpha}
\bibliography{cluster,ref,dpbib}

\section{High-Probability Analysis}
\sloppy
The result from Lemma \ref{lem:decrementalTime} actually holds with high probability (w.h.p.), modulo an extra quasi-linear term in the number $n$ of nodes. 

We call a globally-charged level $i$ \emph{heavy} if $p(i)\geq C\ln n$ for a constant $C=\Theta(\frac{1}{\eps^2})$ to be fixed later, and \emph{light} otherwise. Let also ${\mathcal{H}}$ and ${\mathcal{L}}$ denote the sets of heavy and light levels, respe. We bound the number of tokens charged over light and heavy levels separately. The bound on the light levels is trivial. 
\begin{lemma}\label{lem:timeLight}
Deterministically the number of tokens charged over light levels is at most $O(\frac{f^2}{\eps^3} n\log n)$. 
\end{lemma}
\begin{proof}
Recall that each charged deleted pivot in a globally-charged level $i$ is charged by at most $\frac{4f^2}{\eps}2^i$ tokens. The considered number is therefore deterministically upper bounded by 
\begin{align*}
\sum_{i=0}^{L}Pr[i\in {\mathcal{L}}]\cdot E[\frac{4f^2}{\eps}2^id(i) | i\in {\mathcal{L}}] & \leq \sum_{i=0}^{L}E[\frac{4f^2}{\eps}2^ip(i) | i\in {\mathcal{L}}]\\
& \leq  \sum_{i=0}^{L}\frac{4f^2}{\eps}2^i C\ln n \leq \frac{4f^2}{\eps} 2C n\ln n.
\end{align*}
\end{proof}
Consider next the heavy levels. We slightly modify the adversary such that it maximizes $Pr[\tilde{x}_p(i)\leq X]$ where $X_p(i)=\frac{\eps^4}{256}2^ip$. By precisely the same coupling argument as in Lemma \ref{lem:domination}, one obtains the following.  
\begin{lemma}\label{lem:dominationHP}
$Pr[\tilde{x}_p(i)< X_p(i)]\geq Pr[x(i)< X_p(i) | i\in GC,p(i)=p]$.
\end{lemma}

By applying Chernoff's bound (instead of Markov's inequality), we achieve the following high probability bound on $\tilde{x}_p(i)$ for large enough values of $p$.
\begin{lemma}\label{lem:boundDeletionsHP}
For $p\geq C\ln n$, $Pr[\tilde{x}_p(i) < X_p(i) ]\leq \frac{1}{n^2}$.
\end{lemma}
\begin{proof}
Recall that $E[b(i)]\leq \frac{\eps^2}{16}p$. The random variable $b(i)$ is the sum of independent random variables in $\{0,1\}$. Hence by Chernoff's bound\footnote{Recall that for the sum $B$ of independent $0$-$1$ random variables of expectation $\mu$ and for $\delta\geq 1$, one has $Pr[B> (1+\delta)\mu]\leq e^{-\frac{\delta \mu}{3}}$.}
$$
Pr[b(i)\geq \frac{\eps^2}{8}p]\leq e^{-\frac{1}{3}\frac{\eps^2}{16}p}\leq e^{-\frac{\eps^2}{48}C\ln n}\leq \frac{1}{n^2},
$$
where we imposed $C\geq \frac{96}{\eps^2}$. Given the event ${\mathcal{E}}=\{b(i)< \frac{\eps^2}{8}p\}$, one has that $\tilde{x}_p(i)\geq \frac{\eps^4}{256}2^ip=X_p(i)$ deterministically by the same argument as in Lemma \ref{lem:boundDeletions}. The claim follows.
\end{proof}

\begin{lemma}\label{lem:timeHeavy}
With probability at least $1-\frac{1}{n}$, the number of tokens charged over heavy levels is at most $O(\frac{f^2}{\eps^5}t)$. 
\end{lemma}
\begin{proof}
Similarly to the proof of Lemma \ref{lem:decrementalTime}, for a given heavy level $i$ and $p\geq C\ln n$, one has
\begin{align*}
& Pr[t(i)<X_p(i) | i\in {\mathcal{H}}, p(i)=p] = Pr[t(i)<X_p(i) | i\in GC, p(i)=p] \\
\leq & Pr[x(i)<X_p(i) | i\in GC, p(i)=p] \overset{\text{Lem. \ref{lem:dominationHP}}}{\leq} Pr[\tilde{x}_p(i)<X_p(i)] \overset{\text{Lem. \ref{lem:boundDeletionsHP}}}{\leq} \frac{1}{n^2}.
\end{align*}
It therefore follows that
\begin{align*}
& Pr[t(i)< \frac{\eps^4}{256}2^i p(i) | i\in {\mathcal{H}}]\\
& =\sum_p Pr[p(i)=p | i\in {\mathcal{H}}] Pr[t(i)<X_p(i) | i\in {\mathcal{H}}, p(i)=p]\\
& \leq \sum_p Pr[p(i)=p | i\in {\mathcal{H}}]\frac{1}{n^2}\leq \frac{1}{n^2}.
\end{align*}

Hence by the union bound the probability of the event ${\mathcal{E}}$ that there exists at least one heavy level $i$ with $t(i)<\frac{\eps^4}{256}2^ip(i)$ is at most $\frac{L}{n^2}\leq \frac{1}{n}$. Given that the event ${\mathcal{E}}$ does not happen, by a similar argument as in the proof of Lemma \ref{lem:decrementalTime}, one has that the considered number of tokens is upper-bounded by 
$$
\sum_{i=0}^{L}Pr[i\in {\mathcal{H}} | \overline{{\mathcal{E}}}]\cdot \frac{1024f^2}{\eps^5}\cdot E[t(i) | i\in {\mathcal{H}},\overline{{\mathcal{E}}}]$$
$$\leq
\sum_{i=0}^{L}\frac{1024f^2}{\eps^5}\cdot E[t(i) | \overline{{\mathcal{E}}}]\leq \frac{1024f^2}{\eps^5}\cdot t.
$$ 
\end{proof}  

\begin{lemma}\label{lem:decrementalTimeHP}
The running time of the algorithm in the decremental case is at most $O(\frac{f^2}{\eps^3}n\log n+\frac{f^2}{\eps^5}\cdot t)$ with probability at least $1-\frac{1}{n}$, where $t$ is the number of deletions. 
\end{lemma}
\begin{proof}
Trivially from Lemmas \ref{lem:timeLight} and \ref{lem:timeHeavy}.
\end{proof}
  
The same result applies to the fully-dynamic case as well.

\end{document}